\documentclass{amsart}
\usepackage[margin=3cm]{geometry}

\title{Improving Gebauer's construction of 3-chromatic hypergraphs with few edges}

\author{Jakub Kozik}
\address{Theoretical Computer Science Department, Faculty of Mathematics and Computer Science, Jagiellonian University, Krak\'{o}w, Poland}
\email{Jakub.Kozik@uj.edu.pl}


\thanks{This work was partially supported by Polish National Science Center (2016/21/B/ST6/02165)}

\usepackage{xcolor}

\usepackage{amssymb}
\usepackage{enumitem}
\usepackage[ruled,noline,titlenumbered,linesnumbered,noend,norelsize]{algorithm2e}
\DontPrintSemicolon
\SetNlSty{footnotesize}{}{}
\IncMargin{0.5em}
\SetNlSkip{1em}
\SetKwIF{If}{ElseIf}{Else}{if}{then}{else if}{else}{endif}
\SetKw{Return}{return}
\SetKw{Pick}{pick}
\SetNlSkip{2em}
\newenvironment{algorithm-hbox}{\hbadness=10000\begin{algorithm}}{\end{algorithm}}

\theoremstyle{plain}
\newtheorem{theorem}{Theorem}

\newtheorem{lemma}[theorem]{Lemma}
\newtheorem{corollary}[theorem]{Corollary}

\newtheorem{proposition}[theorem]{Proposition}

\numberwithin{equation}{section}

\renewcommand{\epsilon}{\varepsilon}
\newcommand{\epsi}{\varepsilon}

\renewcommand{\leq}{\leqslant}
\renewcommand{\geq}{\geqslant}

\newcommand{\expt}{\mathbb{E}}

\newcommand{\prob}{\mathrm{Pr}}

\newcommand{\eul}{\mathrm{e}}

\newcommand\vol{\mathit{vol}}
\newcommand\hs{\mathit{HS}}

\begin{document}

%\maketitle
\begin{abstract}
    In 1964 Erd\H{o}s proved, by randomized construction, that the minimum number of edges in a $k$-graph that is not two colorable is $O(k^2\; 2^k)$.
%    $(1+o(1)) \frac{\eul \ln(2)}{4} k^2 2^{k}$ edges are sufficient to build a $k$-graph which is not two colorable.
To this day, it is not known whether there exist such $k$-graphs with smaller number of edges.
Known deterministic constructions use much larger number of edges.
The most recent one by Gebauer requires $2^{k+\Theta(k^{2/3})}$ edges.
Applying derandomization technique we reduce that number to $2^{k+\widetilde{\Theta}(k^{1/2})}$.
\end{abstract}

\maketitle

\newcommand{\dt}{d_\tau}

%events

%constants
\newcommand{\ca}{c_\mathcal{A}}
\newcommand{\cd}{c_\mathcal{D}}

\newcommand{\cp}{c_p}

\newcommand{\sg}{\sigma}

\section{Introduction}
%The main subjects of our considerations are $k$-graphs, i.e. $k$-uniform hypergraphs.
%We denote the number of vertices by $n$, and the number of uniformly sampled edges by $m$.

In 1964 Erd\H{o}s proved in \cite{Erd1964} that $(1+o(1)) \frac{\eul \ln(2)}{4} k^2 2^{k}$ edges are sufficient to build a $k$-graph\footnote{i.e. $k$-uniform hypergraph} which is not two colorable.
To this day that result provides the best known upper bound for the minimum number of edges in such hypergraph.
The Erd\H{o}s' bound results from the fact that random $k$-graph with that number of edges, built on a set of $k^2/2$ vertices 
can not be colored properly with two colors with high probability.

%$H_k(k^2/2,(1+o(1)) \frac{\eul \log(2)}{4} k^2 2^{k})$ can not be colored properly with two colors with high probability.

The best known deterministic construction of a $k$-graph that is not two colorable has been obtained by Gebauer \cite{Gebauer13}.
It requires $2^{k+\Theta(k^{2/3})}$ edges.
It is also the first construction in which the number of edges is $2^{k+o(k)}$.
The main result of the current paper is an upgrade of this construction that allows to cut down the number of edges to $2^{k+ \Theta((k\log(k))^{1/2})}$.

\bigskip

Within the whole paper, $\log(.)$ stands for binary logarithm.
We are only concerned with vertex two coloring of hypergraphs.
Vertex coloring is \emph{proper} if no edge is monochromatic.
Following common convention we use colors \emph{red} and \emph{blue}.
%We also omit a number ceilings ... rounding, as they would ... complications that are irrelevant for the obtained results.

\section{Gebauer's construction}

We start with recalling the construction of \cite{Gebauer13}, as we are going to modify it.
The whole procedure is parametrized by $t=t(k)$ that takes value roughly $k^{\alpha}$ for some optimized positive $\alpha <1$.
It it convenient to organize the vertices of the constructed hypergraph into a rectangular matrix $\mathbb{M}$.
Slightly abusing the notation, we use $\mathbb{M}$ for both the matrix and the set of vertices.
We use the same convention for submatrices of $\mathbb{M}$.
The length of the rows is denoted by $s$. 
%Its specific value will be given later.
Its value will be a subject of optimization.

\subsection{Preliminary choice of rows}
Vertex coloring can be seen as assigning colors to the entries of the matrix.
A color is \emph{dominating} in a row if at least half of its entries are colored with it 
(there can be two dominating colors).
The main part of the construction is designed to work with a submatrix of $t$ rows with the same dominating color.
A matrix for which one of the colors is dominating in all rows will be called \emph{consistently dominated}.
We always assume that red is the dominating color in such a matrix. 

The ground matrix $\mathbb{M}$ has $2t-1$ rows.
Hence, the hypergraph is built on $(2t-1)\cdot s$ vertices.
%Therefore we start with a matrix with $2t-1$ rows, and apply 
Let $\mathcal{M}$ denote the set of submatrices of $\mathbb{M}$ built of every $t$ rows.
%The main construction, described in the next section is applied for submatrices built of every $t$ rows.
For every $M\in \mathcal{M}$ we apply the \emph{main construction} described in the next section.
The construction outputs hypergraph $H_M$.
The union of the edge sets of these hypergraphs forms the edge set of the resulting hypergraph.
%such submatrix we generate some set of hyperedges, and then group them together/take their union to form a resulting hypergraph.
For every coloring of $\mathbb{M}$ at least one submatrix $M \in \mathcal{M}$ is consistently dominated.
The main construction guarantees that in such a case, $H_M$ contains a monochromatic edge.
%among the edges generated from that submatrix.

\subsection{Main construction}
Let $M\in \mathcal{M}$, recall that $M$ has $t$ rows.
Our goal is to build a hypergraph $H_M$ on the vertex set $M$ such that for every consistently dominated coloring of $M$, there exists a monochromatic edge in $H_M$.
%, and we are concerned only with colorings in which the matrix is consistenly dominated.
%in which all of them have the same dominating color.
%We assume wlog that red is the dominating color of the rows.
For $(\sg_1, \ldots, \sg_t) \in [s]^t$, we denote by $M(\sg_1, \ldots, \sg_t)$ matrix $M$ in which for every $i\in [t]$, 
the $i$-th row has been cyclically shifted by $\sg_i$.
The construction proceeds as follows.
\bigskip
\noindent \\
\texttt{
For every
\begin{enumerate}
    \item \emph{sequence of shifts} $\sg \in [s]^t$,
    \item and \emph{set of indices} $I \subset [s]$ of size $k/t$,
\end{enumerate}
add to $H_M$ an edge built from all elements of the columns of $M(\sg)$ with indices in $I$.
}
\bigskip

Note that the edges of $H_M$ are of size $k$ as required.
%in the matrix in which rows are cyclically shifted according to sequence $(s_1, \ldots, s_t)$
%(i.e. , for every $i\in [t]$, the $i$-th row has been cyclically shifted by $s_i$).

Let us fix a consistently dominated coloring of $M$.
We assume wlog that red is the dominating color of the rows.
When the sequence of shifts is chosen randomly, the probability that some fixed column is red is at least $2^{-t}$.
As a consequence, for $s \geq (k/t) \; 2^t$ the expected number of red columns is at least $k/t$.
In particular, for some sequence of shifts, there exists a set of $k/t$ red columns. 
Hence the edge built for these shifts and columns is monochromatic.

\subsection{Counting}
We have 
\[
    {2t-1 \choose t} < 2^{2t}
\]
choices for the subset of rows in the preliminary step.
Then, in the main construction, every sequence of $t$ elements of $[s]$ and a subset of $k/t$ elements of $[s]$ is used to build an edge.
The number of choices is
\[
    s^t \cdot {s \choose k/t } \leq s^t \cdot \left( \frac{\eul s }{ k/t } \right)^{k/t}.
\]
For $s = (k/t) \; 2^t$ (we assume for simplicity that it is an integer) we obtain
\[
    (k/t)^t\; 2^{t^2} \cdot  \eul^{k/t} \; 2^k = 2^{t \log(k/t) + t^2 + k/t \log(\eul) + k }.
\]
The total number of edges is smaller than
\[
        2^{2t+ t \log(k/t) + t^2 + k/t \log(\eul) + k }.
\]
Finally we choose $t$ so that the above exponent is minimized. %$t+ t \log(k/t) + t^2 + k/t \log(\eul)$. 
That happens for $t = \Theta( k^{1/3} )$.
In the end we obtain that the total number of edges is $2^{k+ \Theta(k^{2/3})}$.

\newcommand{\cc}{\mathcal{C}}
\section{Improved construction}
We modify only the main construction.
Recall that we work with matrix $M$ with $t$ rows.
For a fixed consistently dominating coloring of $M$, sequence of shifts $\sg \in [s]^t$ is called \emph{good} if $M(\sg)$ contains at least $s\;2^{-t}$ red columns.
The set of good sequences for a coloring $\cc$ of $M$ is denoted by $\mathcal{G}(\cc)$.

If we fix a consistently dominating coloring of $M$ and choose the sequence of shifts $\sg \in [s]^t$ uniformly at random,
the expected number of red columns in $M(\sg)$ is $s\;2^{-t}$.
That observation was used to justify that there exists a good sequence.
%Value of $s$ was chosen to be such that the expected number was at least $k/t$. 
However, it also suggests that a large number of shift sequences might be good.
For the constructed hypergraph not to be two colorable, it is sufficient that for every consistently dominated coloring of $M$,
at least one such sequence is used in the main construction.

We apply derandomization techniques to construct relatively small set of sequences of shifts that can be used in the main construction instead of $[s]^t$.
For a family of sets $\mathcal{F}$, a set that intersects every element of that family is called a \emph{hitting set} for $\mathcal{F}$.
In these terms we are looking for a small hitting set for family $\mathcal{G}_M = \{\mathcal{G}(\cc) : \text{$\cc$ is a consistently dominating coloring of $M$} \}$.

%Imagine that we have $t$ rows of length $s$ with the same dominating color 
%(assume that \emph{red} is dominating).
%Let the sequence of shifts $\sg \in [s]^t$ be chosen u.a.r. (uniformly at random).
%Random variable $X$ is defined as the number of red columns in the shifted matrix $M(\sg)$.
%By the previous argument (and for the chosen value of $s$) we have
%\[
%    \expt[X]= s \; 2^{-t}= k/t.
%\]
%That suggests not only that for some $\sg$ there are at least $k/t$ red columns,
%but also that large number of shift sequences leads to such situation.
%For the hypergraph not to be two colorable it is sufficient that for every consistently dominated coloring of $M$ 
%at least one such sequence is used in the main construction.

\subsection{Sequential choice of shifts}
\label{sec:SeqChoice}
We start with estimating the size of the set of good shift sequences.
While it is not directly used in our construction, it provides good opportunity to introduce some tools.
It will also allow to derive a probabilistic argument that small hitting sets actually exist.

The property of being good is generalized to prefixes in the straightforward way --
sequence of shifts $(\sg_1, \ldots, \sg_i)$ is \emph{good} if the matrix trimmed to the first $i$ rows and shifted according to the sequence,
has at least $s \; 2^{-i}$ red columns.

Suppose that $(\sg_1, \ldots, \sg_i)$ is good. %in such a way that in the matrix trimmed to the first $i$ rows, there are at least $s\; 2^{-i}$ red columns.
We want to estimate the number of possible choices of $\sg_{i+1}$ for which $(\sg_1, \ldots,\sg_i, \sg_{i+1})$ is good as well. 
%the matrix trimmed to the first $i+1$ rows has at least $s\; 2^{-(i+1)}$ red columns.
If the coloring of the $(i+1)$-th row was "random", then about half of the choices would be right, and almost all of the choices would be almost right.
That property does not hold in the worst case scenario and hence we are going to work with relaxed definitions.

For $\epsi>0$, a sequence of shifts $(\sg_1, \ldots, \sg_i)$ is \emph{$\epsi$-good} if the number of red columns in the shifted matrix trimmed to the first $i$ rows is at least $ s\; \left(\frac{1-\epsi}{2} \right)^{i} $.
Then, every $\epsi$-good sequence of shifts of length $t$ gives a shifted matrix with at least
\[
    s \frac{(1-\epsi)^{t-1}}{2^t}  
\]
red columns.
For $s \geq \eul \; (k/t) \; 2^{t}$ and $\epsi= 1/t$, the number of red columns is at least $k/t$ as needed.
%... We are going to work with slightly alterd parametrization -- we put $s = \eul (k/t) 2^{t}$.
In the modified construction we set $s$ to $\lceil \eul \; (k/t) \; 2^{t} \rceil$.

We also define $\mathcal{G}^\epsi(\cc)$ as the set of $\epsi$-good sequences for a coloring $\cc$ of $M$
and $\mathcal{G}^\epsi_M$ as $ \{\mathcal{G}^\epsi(\cc) : \text{$\cc$ is a consistently dominating coloring of $M$} \}$.

The following proposition is used to derive a lower bound for the number of $\epsi$-good sequences.
It is formulated in more general terms that needed here, but we are going to use it again later. %/ for the sake of future//further use.
For a set $A\subset[s]$ and a number $x$, set $A+x$ is defined as the set $A$ shifted cyclically within $[s]$ by $x$, 
    formally $A+x=\{ (a-1+x) (\mod s) +1 : a\in A \}$.
Purely technical proof of the proposition is moved to Appendix \ref{app:nextShiftVolume}.

%    (where $A+x = \{ (a+x) (\mod s) +1 : a\in A \}$)
\begin{proposition}
\label{prop:nextShiftVolume}
    For any positive $\epsi <1$ and sets $A, B\subset [s]$, let $\alpha = |B| / s$,
    there exist at least 
    \[
        \frac{\epsi}{1- (1-\epsi)\alpha} \;\alpha s
%        \frac{\epsi}{1+ \epsi} \; \alpha s
    \]
    elements $x \in [s]$ for which $|(A+x) \cap B | \geq (1-\epsi) \alpha |A|$.
%    (where $A+x = \{ (a+x) (\mod s) +1 : a\in A \}$)
\end{proposition}

For $|B|\geq s/2$ we get that there exist at least 
    \[
        \frac{2 \epsi}{1+ \epsi} \; s/2
    \]
elements $x \in [s]$ for which $|(A+x) \cap B | \geq (1-\epsi) |A|/2$.

Applying the proposition iteratively, we obtain that the number of $\epsi$-good sequences of length $j$ is at least
\[
    \left( \frac{\epsi}{1+ \epsi} \; s \right)^{j}.
\]
(For a fixed $j$, and some $\epsi$-good sequence $\sg$ of length $j-1$,  
    let $A$ be the set of indices of the red columns in the matrix trimmed to the first $j-1$ rows and shifted according to $\sg$,
    and $B$ be the set of indices of red entries of the $j$-th row.)

For $j=t$ we get a lower bound for the number of $\epsi$-good sequences.
Once we have that bound, 
%for the volumes of the set of $\epsi$-good shifts, 
typical application of the probabilistic method (along the lines of the proof from \cite{Erd1964}) 
allows to proof that there exists a hitting set for $\mathcal{G}^\epsi_M$ of size $2^{O(t \log(t))}$ 
 %uniformly random set of $2^{O(t \log(t))}$ elements of $[s]^t$, with high probability, is a hitting sets for $\mathcal{G}_M$.
(see Appendix \ref{app:hs_prob}).
We are interested however in deterministic construction.

%\begin{remark}
%    In the above discussion we do not use any properties of shifts, beside the fact that a random shift moves any fixed element to a uniformly random position.
%    It does not also gain any benefits from shifts which leave more that half red columns in tact.
%    Clearly there is some space for improvements, but we need bigger picture to understand what improvements would be significant.
%    It is tempting to organize set $[s]$ into a vector space and use some vector operations instead of shifts.
%\end{remark}

\subsection{Expanders for hitting sets}
Linial, Luby, Saks and Zuckerman \cite{llsz} worked on deterministic constructions of small hitting sets for combinatorial rectangles.
We summarize in this section, their results that are relevant for our developments.
We follow closely their definitions.

Graph $G=(V,E)$ is an \emph{$(m,\Delta, \alpha)$-expander} if it has $m$ vertices, maximum degree $\Delta$ 
and for any $A\subset V$,
the fraction of vertices in $V - A$ that have a neighbor in $A$ is at least $\alpha |A|/m$.
For a fixed graph $G$ let $W_r$ denote the set of walks in $G$ of length $r$.
Let $W_{r,d}$ be the set of subsequences of elements of $W_r$ of length $d$ 
%sequences of vertices of length $d$, which are subsequences of elements of $W_r$
(not necessarily subsequences of consecutive elements).
Set $R\subset [m]^d$ is a \emph{combinatorial rectangle} if it is of a form $R_1\times \ldots \times R_d$ for some $R_1, \ldots, R_d\subset [m]$.
The volume of rectangle $R$, denoted as $\vol(R)$, is defined as $|R|/m^d$.
\begin{lemma}[\cite{llsz}]
\label{lem:llsz}
    Let $m,d$ be positive integers and $R$ be a rectangle in $[m]^d$.
    Suppose $G$ is an $(m, \Delta , \alpha)$-expander with $1/2 > \alpha > 0$. 
    If $r = 1 + (4/ \alpha) (d + \log(1/\vol(R)))$, then $W_{r,d}$ contains a point from $R$.
\end{lemma}

The above lemma implies that a specific set of sequences $W_{r,d}$ hits every combinatorial rectangle in $[m]^d$ of sufficiently large volume.

The following rough estimations for the size of $W_{r,t}$ will be sufficient for our needs.
We have 
\[
       | W_r | \leq m (\Delta+1)^r
\]
and 
\[
    |W_{r,d}| <  2^r |W_r| \leq m\; (2(\Delta+1))^{r}.
\]

Lemma \ref{lem:llsz} leaves some space for the choice of expander graph.
Authors of \cite{llsz} used the construction of Margulis \cite{Margulis} (see also \cite{GabberGalil}) which allows to build an expander with $\Delta =8$ and $\alpha=(2-\sqrt{3})/4$.
A minor inconvenience is that the construction requires the number of vertices to be a perfect square.
However, as observed already in \cite{llsz}, we can consider the rectangles of our interest as subsets of a larger space $[m']^d$, and apply the lemma in that space.
For every $m$ we can choose number $m'$ that is a perfect square and satisfies $m \leq  m' \leq 2 m$.
% that bound can be improved for large $m$, we do not need this. At the same time we apply the lemma for small m.
While that change affects the volumes of rectangles, they get smaller at most by a factor of $2^{-d}$. 
For our purposes this cost is negligible.

When we are interested in rectangles of volume at least $\mathcal{V}$, Lemma \ref{lem:llsz} instructs to take
\[
    r = r(d, \mathcal{V})  = 1 + (4/ \alpha) (d + \log(2^{d}/\mathcal{V})).
\]
For some specific constant $\hat{C}$ and for all positive $d$ and $\mathcal{V}$ we have 
%There exists an absolute constant .... (i.e. the same for all volumes and dimensions) for which we have
\[
    r(d, \mathcal{V})  \leq \hat{C} (d+\log(1/\mathcal{V})).
\]

\begin{corollary}
\label{cor:hittingSet}
    There exists constant $C>0$ such that, 
    for every integers $m,d$, and $\mathcal{V}>0$ there exists a subset of $[m]^d$ of size at most % $HS\in [m]^d$ of size at most
    \[
        %2 m \cdot 18^{C (d+\log(1/\mathcal{V}))} = 
        m \cdot 2^{C(d+ \log(1/\mathcal{V}))},
    \]
    that intersects every combinatorial rectangle in $[m]^d$ of volume at least $\mathcal{V}$. 
\end{corollary}

We apply that result, to construct a small hitting set for $\mathcal{G}^\epsi_M$.
That set is then used in the modified main construction instead of the set of all shift sequences.
%$t$ shifts that give at least $k/t$ red columns in the shifted matrix.

\subsection{Under false assumption}
Unfortunately, for a fixed consistently dominating coloring of $M$, the set of good or $\epsi$-good shift sequences does not need to form a combinatorial rectangle.
It is instructive to pretend for a moment that it does.
We assume (falsely) in this subsection that $\mathcal{G}^\epsi_M$ contains only combinatorial rectangles.
%For a fixed  consistent coloring of $M$ (i.e. red is dominating in each of $t$ rows) the volume of the set of $\epsi$-good shifts in $[s]^t$ is can be derived from 

By the discussion that follows Proposition \ref{prop:nextShiftVolume}, for every consistently dominating coloring of $M$, the set of $\epsi$-good shift sequences has volume at least 
\[
  \nu =  \left( \frac{\epsi}{2(1+ \epsi)} \right)^t.
\]
By Corollary \ref{cor:hittingSet} there exists a hitting set $\hs$ for all rectangles of volume $\nu$ of size $s\cdot 2^{C(t+ \log(1/\nu))}$.
For $\epsi = 1/t$ and $s= \lceil \eul\; (k/t) \; 2^{t} \rceil$, the size of $\hs$ is at most $2^{2C t \log(t)}$ (assuming that $t$ is sufficiently large).
%\[
%    s\cdot 2^{O(t+ \log(1/\nu))} = s \cdot 2^{O(t \log(t))}.
%\]
%For $\epsi = 1/t$ we get that $r= O(t \log(t))$ and since $s=2^{O(t)}$ we also have
%\[
%    W_{r,s}  = 2^{O(t \log(t))}.
%\]
Note that in the original construction all possible shift sequences were used. 
%Their number is $s^t = 2^{O(t^2)}$.
Using set $\hs$ instead of $[s]^t$ %in the main construction, 
and choosing $t= (k \log(k))^{1/2}$,
the total number of edges becomes
\[
    2^{k+ O((k \log(k))^{1/2})}.
\]

%
%Using $HS$??? instead of the set of all shifts, the total number of edges built in the main construction becomes
%\[
%    2^{O(t \log(t))} \cdot {s \choose k/t } \leq 2^{O(t \log(t))} \cdot \left( \frac{\eul s }{ k/t } \right)^{k/t}.
%\]
%With $s = \eul \; (k/t) \; 2^t$ we obtain
%\[
%    2^{O(t \log(t))} \cdot \eul^{2 k/t}\; 2^k = 2^{O (t \log(t)) + 2 k/t \log(\eul) + k }.
%\]
%The total number of edges is 
%\[
%    2^k \cdot 2^{2t+ O(t \log(t)) + 2 k/t \log(\eul) }.
%\]
%This time it optimizes for $t= \Theta( (k/\log(k) )^{1/2})$.
%The resulting number of edges is of the order
%\[
%    2^{k+ O((k \log(k))^{1/2})}.
%\]

\newcommand{\dd}{h}
\subsection{Decomposing good shift sequences}
We showed in Section \ref{sec:SeqChoice} that, for every consistently dominating coloring of $M$, the set of $\epsi$-good shift sequences is large.
While, in general, it does not have a structure of combinatorial rectangle,
in some sense it can be decomposed into a small number of such.
We start by altering the way that the sequences of shifts are represented.
For the clarity of the exposition we assume that $t$ is a power of 2.

Let $T$ be a rooted plane complete binary tree with $t$ leaves\footnote{
    i.e. all the internal nodes of $T$ have two children (left and right) and all the leaves are of the same distance from the root
}.
%We direct the edges of $T$ from the root to the leaves.
A subtree rooted at some internal node of $T$ consists of that node and all its descendants.
A node of $T$ is \emph{at level $j$} if its distance to the set of leaves is $j$.
Let $S_j$ be the set of inner nodes at level $j$.
Note that $|S_j|= t \;2^{-j}$, we denote that value by $d_j$.
For $\dd= \log(t)$, the tree has $\dd+1$ levels with all the leaves on level 0.

We associate leaves of $T$ with rows of $M$ in such a way that the $i$-th leaf from the left, corresponds to the $i$-th row.
Inner nodes of the tree are going to be labeled by elements of $[s]$.
These labels represent the relative shifts between neighboring rows of $M$.
For an inner node $v$, if $l$ is the rightmost leaf of the left subtree of $v$ and $r$ is the leftmost leaf of the right subtree of $v$, 
then the label of $v$ describes how row $r$ is shifted wrt $l$. %the relative shift between rows $l$ and $r$.
%Note also that these are neighboring rows of the matrix.
%that when the labels of the nodes are written in the in-order manner, 
%then the label of the inner node that is just between the labels of two rows of the matrix denotes the relative shift between these rows.
%In other words, 

Labeling of a subtree rooted at node $v$ is \emph{$\epsi$}-good, if 
    for $r$ being the number of descendant leaves of $v$,
    the submatrix of the rows that correspond to these leaves, 
        shifted according to the labels of the inner nodes of the subtree, 
    has at least $s \; ((1-\epsi)/2)^{r}$ red columns.
%Labeling of the first $j$ levels of $T$ is \emph{$\epsi$-good} if 
%    for every node $v$ of level at most $j$ if $r$ is the number of descendant leaves of $v$,
%    then the submatrix of $r$ rows corresponding to these leaves, shifted according to the labels of the inner nodes of the subtree rooted at $v$, 
%        has at least $s \; ((1-\epsi)/2)^{r}$ red columns.
%
Note that $\epsi$-good labellings of the whole tree correspond to $\epsi$-good sequences
(up to a cyclic shift of the whole matrix, which is clearly redundant in the original construction).
%A plane binary tree with $d$ leaves and inner nodes labelled by elements of $[s]$ is an \emph{$\epsi$-good shift tree} if 
%    for every node of the tree with $r$ descendant leaves,
%    the submatrix of $r$ rows corresponding to the leaves, shifted according to the labels of the inner nodes of the subtree, 
%        has at least $s \; ((1-\epsi)/2)^{-r}$ red columns.

%\subsubsection{Construction}

We order the nodes of $S_j$ from left to right and represent labellings of the nodes of $S_j$ as elements of $[s]^{d_j}$.
We are going to work bottom up and label inner nodes in groups consisting of the nodes of the same level.
A labeling of $T$ is \emph{$\epsi$-good up to level $j$} if all the subtrees rooted at level at most $j$ are $\epsi$-good.
In all the places where we use this definition, it can be assumed that the labeling is undefined for the nodes of higher levels.
Suppose that $\tau$ is a labeling of $T$ that is $\epsi$-good up to level $j-1$.
Then, a sequence of labels $\sg\in [s]^{d_j}$ is called \emph{an $\epsi$-good level $j$ extension (of $\tau$)}
if the labeling $\tau$ in which the labels of the nodes of level $j$ has been set to $\sg$ is $\epsi$-good up to level $j$.

\begin{proposition}
\label{prop:level_good_labeling}
    Suppose, that a labelling of $T$ is $\epsi$-good up to level $j-1$.
%    the first $j-1$ levels of $T$ are labelled in such a way, 
%    that all the subtrees rooted at level $j-1$ are $\epsi$-good. 
    Then, the set of its $\epsi$-good level $j$ extensions
%    sequences of labels for the nodes of level $j$, such that the labeling extended by these values
%        is $\epsi$-good up to level $j$ 
%    make subtrees rooted at level $j$ $\epsi$-good,
    forms a combinatorial rectangle of volume at least
\[
    \nu_j = \left( \epsi \; ((1-\epsi)/2)^{-2^{j-1}} \right)^{d_j}.  
%    \epsi^{t 2^{-j}} ((1-\epsi)/2)^{t/2}
\]
\end{proposition}
\begin{proof}
Fix $j$ and suppose that labeling $\tau$ is $\epsi$-good up to level $j-1$.
    %every subtree rooted at a node of level $<j$ is an $\epsi$-good shift tree.
We want to assign labels to the nodes of $S_j$ in such a way that all the subtrees rooted at depth $j$ are $\epsi$-good shift trees as well.
%Since the labels in all subtrees rooted at levels $>j$ are already fixed, we need to choose a good? shift for every node of $S_j$.
Note that for any pair of distinct nodes of level $j$, the property of the corresponding subtrees of being $\epsi$-good shift trees 
are determined by disjoint sets of rows of the underlying matrix.
    %therefore their values are independent.
    %These values of shifts at fixed level concern disjoint sets 
That justify that the set of $\epsi$-good level $j$ extensions forms a combinatorial rectangle.
%In order to ... a small hitting set we estimate its volume.

Let $v$ be a node of $S_j$ and let $A$ and $B$ be the sets of indices of red columns respectively 
    in the shifted submatrices corresponding to the left and right subtrees of $v$.
By the assumptions we know that both these sets have cardinality at least 
\[
    s \; ((1-\epsi)/2)^{-2^{j-1}}.
\]
We need to estimate the number of $x\in [s]$ for which the set $A \cap (B +x)$ has cardinality at least
\[
    s \; ((1-\epsi)/2)^{-2^{j}}.
\]
Proposition \ref{prop:nextShiftVolume} gives that there exist at least 
\[
    \epsi \; ((1-\epsi)/2)^{-2^{j-1}} \; s
\]
such values.
We obtain that the volume of combinatorial rectangle of $\epsi$-good level $j$ extensions is at least 
\begin{align*} 
    \left( \epsi \; ((1-\epsi)/2)^{-2^{j-1}} \right)^{d_j}
%    \\ 
%    &= \left( \epsi ((1-\epsi)/2)^{-2^{j-1}} \right)^{t 2^{-j}} \\
%    &= \epsi^{t 2^{-j}} ((1-\epsi)/2)^{-2^{j-1}t 2^{-j}} \\
%    &= \epsi^{t 2^{-j}} ((1-\epsi)/2)^{t/2} \\
\end{align*}
\end{proof}
%
%By Theorem ... there exist a hitting set $H_j$ for all such rectangles of size ...
%For every $j$ we define $r_j$ as follows
%\begin{align*}
%    r_j & = 1 + (4/ \alpha) (d_j + \log(1/vol(G_j)))\\
%        & = 1 + (4/ \alpha) (t \; 2^{-j} + t 2^{-j} \log( 1/\epsi) + t/2 \log ((1-\epsi)/2)).
%\end{align*}

By Corollary \ref{cor:hittingSet}, there exists a set $\hs_j$ of cardinality
\[
    s \cdot 2^{C (d_j + \log(1/\nu_j))},
\]
that is a hitting set for the family of $\epsi$-good level $j$ extensions for labellings that are $\epsi$-good up to level $j-1$.
That implies the following proposition.

\begin{proposition}
\label{prop:finalHS}
    Set $\hs = \hs_1 \times \ldots \times \hs_\dd$ is a hitting set for the family of sets of $\epsi$-good labellings of $T$.
\end{proposition}
It remains to estimate the size of $\hs$.
We have 

\begin{align*}
    |\hs| &\leq \prod_{j=1\ldots \dd} s \cdot 2^{C (d_j + \log(1/\nu_j))} \\
    &= s^{\log(t)} \cdot 2^{C \sum_{j=1\ldots \dd} (d_j + \log(1/\nu_j)) } \\
    &< s^{\log(t)} \cdot 2^{C t} \cdot 2^{C \sum_{j=1\ldots \dd} \log(1/\nu_j) },
\end{align*}
and
\begin{align*}
    \sum_{j=1\ldots \dd} \log(1/\nu_j) 
        &= \sum_{j=1\ldots \dd} d_j (\log(1/\epsi) +2^{j-1} \log(2/(1-\epsi)) ) \\
        &< t \log(1/\epsi) + t \sum_{j=1\ldots \dd} \log(4) \;\;\;\;\text{(for $\epsi<1/2$)}\\
        & = t \cdot \log(1/\epsi) + 2t \cdot \log(t)  
\end{align*}

Therefore, for our parametrization (i.e. $s= \lceil \eul\; (k/t)\; 2^t \rceil$ and $\epsi= 1/t$), and for all sufficiently large $t$ we get
\[
    |\hs| \leq 2^{4 t \log(t)}.
\]

\subsection{Modified main construction}
Let $\hs$ be the set from Proposition \ref{prop:finalHS}.
As we already observed labellings of $T$ correspond to shift sequences up to a cyclic shift of the whole matrix.
For a labeling $\tau$ let $\sg(\tau)$ be a shift sequence that is compatible with $\tau$.
Observe, that if $\tau$ is an $\epsi$-good labeling, then $\sg(\tau)$ is $\epsi$-good shift sequence.
Recall that we chose $s= \eul\; (k/t)\; 2^t$ so that if $\sg$ is an  $\epsi$-good sequence for some consistently dominated coloring of $M$,
then $M(\sg)$ has at least $k/t$ red columns.
The modified main construction proceeds as follows.

\texttt{
\noindent\\
For every
\begin{enumerate}
    \item labeling of the tree $\tau \in \hs$,
    \item and set of indices $I \subset [s]$ of size $k/t$,
\end{enumerate}
add to $H_M$ an edge build from all elements of the columns of $M(\sg(\tau))$ with indices in $I$.
}

By Proposition \ref{prop:finalHS} for every consistently dominated coloring of $M$, at least one $\epsi$-good labeling $\tau$ is used in the construction.
Then, for every such coloring, matrix $M$ shifted according to $\sg(\tau)$ has at least $k/t$ red columns.
As a consequence at least one of the edges of $H_M$ is monochromatic.

\subsubsection{Counting}

Just like in the original construction, we have less than $2^{2t}$
choices for the subset of rows in the preliminary step.
Then, in the modified main construction, we use every sequence of $\hs$ with every subset of $k/t$ elements of $[s]$ to build an edge.
The number of choices is smaller than
\[
    2^{4 t \log(t)} \cdot {s \choose k/t } < 2^{4 t \log(t)} \cdot \left( \frac{\eul s }{ k/t } \right)^{k/t}.
\]
Substituting the value of $s$ we obtain a value that is smaller than
\[
    2\cdot 2^{ 4 t \log(t)} \cdot  \eul^{2k/t} 2^k = 2^{1+ 4 t \log(t) + (2k/t) \log(\eul) + k }.
\]
The bound is multiplied by $2$ to compensate for the ceiling in the definition of $s$.
Taking into account preliminary choices of rows, the total number of edges is smaller than
\[
        2^{2t+ 1+  4 t \log(t) + (2k/t) \log(\eul) + k}.
\]
For $t=(k/ \log(k))^{1/2}$, the total number of edges becomes $2^{k+ \Theta((k\log(k))^{1/2})}$.

%This time the exponent is minimized for $t = \Theta( ki...^{1/2} )$.
%In the end we obtain that the total number of edges is $2^{k+ \Theta(k..^{1/2})}$.
%%
%% Bibliography
%%

%% Please use bibtex, 

\bibliographystyle{amsplain}
\bibliography{pb}

\appendix

\section{Proof of Proposition \ref{prop:nextShiftVolume}}
\label{app:nextShiftVolume}
\begin{proof}
    Let random variable $X$ denote the size of $(A+x) \cap B$, when $x\in [s]$ is chosen uniformly at random.
    By the fact that $|B|= \alpha s$ and linearity of expectation we obtain
    \[
        \expt(X) = \alpha |A|.
    \]
    From the definition of $X$, we get also
    \[
        X \leq |A|.
    \]
    We can observe now that a distribution that minimizes $\prob[ X > (1-\epsi) \alpha |A| ]$ and satisfies the above conditions, 
        is supported only by values $(1-\epsi) \alpha |A|$ and $|A|$.
    There is only one such distribution that satisfies $\expt(X)= \alpha |A|$. % already determines the distribution.
    Straightforward calculations give
    \[
        \prob[ X > (1-\epsi) \alpha |A| ] \geq \frac{\epsi \alpha}{1- (1-\epsi)\alpha}.
    \]
%    This can be proved by observing that the worst possible distribution of $X$ would be concentrated on values $(1-\epsi) \alpha |A|$ and $|A|$.
%    There is only one such distribution that satisfies $\expt(X)= \alpha |A|$.% already determines the distribution.
%    \emph{The rest is mathematics ...}
\end{proof}

\section{Small hitting sets exist}
\label{app:hs_prob}
Recall that, for a fixed consistently dominated coloring of matrix $M$ with $t$ rows, the volume of $\epsi$-good sequences is at least
\[
   p=  \left( \frac{\epsi}{1+ \epsi} \right)^{t}.
\]
The volume is exactly the probability that uniformly random sequence is $\epsi$-good.
Let $S$ be a set built from $m$ uniformly and independently sampled random sequences from $[s]^t$.
(Since, the sequences are sampled with repetitions, it may happen that $|S|<m$.)
% (We use multiset for technical convenience of choosing with repetitions... allows for sampling sequence of edges with repetitions ... independent sampling)
The following formula upperbounds the expected number of consistently dominated colorings of $M$, for which the set of $\epsi$-good sequences is not hit by $S$
%instances??? for which $\epsi$-good sequences are not hit by the random set of sequences 
\begin{align*}
    2^{st}  \cdot (1-p)^m 
        &<  \exp(st \ln(2) -  m p).
\end{align*}
Therefore, whenever $s t \ln(2) -  m p \leq 0$, some set of $m$ sequences hits all the sets of $\epsi$-good sequences for consistently dominating colorings.
%Solving ... for $m$ we obtain $ m \geq s t \ln(2)/ p $.
For $s= \lceil \eul\; (k/t)\; 2^t \rceil$ and $\epsi = 1/t$ it is sufficient to take $m$ of the order $2^{O(t \log(t))}$ to satisfy the inequality.
As a consequence there exists a hitting set for $\mathcal{G}^\epsi_M$ of size $2^{O(t \log(t))}$.

\end{document}